\documentclass[lettersize,journal]{IEEEtran}
\usepackage{amsmath,amsfonts}
\usepackage{amsthm}
\usepackage{amssymb}
\usepackage{epstopdf}
\usepackage{algorithmic}
\usepackage{algorithm}
\usepackage{array}
\usepackage[caption=false,font=normalsize,labelfont=sf,textfont=sf]{subfig}
\usepackage{textcomp}
\usepackage{stfloats}
\usepackage{url}
\usepackage{verbatim}
\usepackage{graphicx}
\usepackage{cite}
\newtheorem{prop}{Proposition}
\newtheorem{lem}{Lemma}
\newtheorem{coro}{Corollary}
\newtheorem{rem}{Remark}

\begin{document}

\title{Rate Optimization for Downlink URLLC via Pinching Antenna Arrays}

\author{Tong Lin, Jianyue Zhu, Wei Huang, Meng Hua, Zhizhong Zhang
\thanks{
	Tong Lin, Jianyue Zhu, and Zhizhong Zhang are
	with the College of Electronic and Information Engineering, Nanjing
	University of Information Science and Technology, Nanjing 210044, China
	(e-mail: 202412490664@nuist.edu.cn; zhujy@nuist.edu.cn; zhangzz@nuist.edu.cn).
	
	Wei Huang is with the School of Computer Science and Information
	Engineering, Hefei University of Technology, Hefei 230009, China (e-mail:
	huangwei@hfut.edu.cn).
	
	M. Hua is with the Department of Electrical and Electronic Engineering, Imperial College London, London SW7 2AZ, UK (e-mail: m.hua@imperial.ac.uk).}
}



\maketitle

\begin{abstract}
This work studies an ultra-reliable and low-latency communications (uRLLC) downlink system using pinching antennas which are realized by activating small dielectric particles along a dielectric waveguide. Our goal is to maximize the data rate by optimizing the positions of the pinching antennas. By proposing a compact and cost-efficient antenna architecture and formulating a finite blocklength-based optimization model, we derive a closed-form solution for the optimal antenna placement under quality-of-service (QoS) and antenna spacing constraints.
Meanwhile, a phase-alignment strategy is integrated into the design, enabling coherent signal superposition across the array. 
Simulation results confirm significant rate improvements over conventional antenna systems while satisfying uRLLC requirements, making the proposed design well-suited for compact and latency-critical future applications. 
\end{abstract}

\begin{IEEEkeywords}
Pinching antenna, ultra-reliable and low-latency communications (uRLLC), antenna position optimization.
\end{IEEEkeywords}

\section{Introduction}
\IEEEPARstart{I}{n} recent years, the demand for compact, energy-efficient, and high-performance wireless systems has driven the development of mechanically flexible and electrically reconfigurable antennas. Flexible and soft-body antennas, made from stretchable or fabric-based materials, have shown a promising solution to meet the demands of next-generation wireless platforms such as reconfigurable intelligent surfaces (RISs), fluid-antenna systems and movable antennas \cite{huang2019reconfigurable,wong2020fluid,zhu2023modeling}. These antennas not only enable the user's wireless channel to be reconfigured as a system parameter, but also offer adaptability and ease of integration into non-rigid environments. However, traditional flexible antennas suffer from large-scale path loss, hindering their use in high-performance settings \cite{ding2025flexible}.

To address these issues, pinching antennas have been proposed as a low-efficient waveguide-fed solution that eliminates the need for dedicated radio frequency (RF) chains for each antenna. This approach introduces a dielectric waveguide, allowing each antenna to passively extract energy from the guided signal \cite{ ouyang2025array}.
Unlike conventional fixed-position antennas, pinching antennas support flexible deployment without requiring extra hardware, offering a cost-effective approach for establishing adaptable line-of-sight (LoS) links \cite{ding2025flexible}.
This architecture further simplifies RF design, supports compact multi-antenna integration with phase-coherent transmission, and is also low-cost and easy to deploy since antenna configuration only requires inserting or removing dielectric materials. 

Recent studies have advanced both the theoretical foundations and practical implementations of pinching antenna systems.
Initial research has primarily focused on electromagnetic modeling and array design. Yang et al. \cite{yang2025pinching} presented a comprehensive survey on pinching antennas, discussing their electromagnetic behavior, integration with dielectric waveguides, and potential deployment in future 6G networks.
Ouyang et al. \cite{ouyang2025array} derived closed-form expressions for the achievable array gain  for pinching antennas, offering design insights on optimal element spacing in linear pinching arrays.

Based on these foundations, researchers have explored integrating pinching antennas with emerging wireless technologies. Flexible architectures employing pinching antennas in NOMA systems have been shown to achieve near-theoretical capacity under interference-limited conditions \cite{ding2025flexible}. Zhu et al. \cite{zhu2023modeling} explored the collaborative deployment of pinching antennas and movable antennas. A comparative study has been conducted between the pinching antenna system and RIS operating in the millimeter-wave band \cite{samy2025pinching}. Meanwhile, leveraging the low-cost and reconfigurable characteristics of pinching antennas, it is possible to improve positioning accuracy while achieving flexible user-centric positioning to support integrated sensing and communications (ISAC) \cite{ding2025pinching}.

In addition, ultra-reliable low-latency communication (uRLLC) is a critical application scenario in 5G and the upcoming 6G, targeting 99.999\% reliability and sub-millisecond latency \cite{li20185g, mahmood2020predictive}. Applications such as factory automation and autonomous driving, impose stringent quality-of-service (QoS) requirements that fundamentally deviate from classical Shannon-theoretic assumptions of infinite blocklength and negligible decoding error probability \cite{polyanskiy2010channel}. Despite the critical importance of uRLLC, the potential of pinching antenna architectures remains largely unexplored.


In this work, we study a single-user downlink uRLLC system employing waveguide-fed pinching antennas. A compact and cost-efficient antenna architecture is proposed and the rate-maximization problem is formulated under a minimum spacing constraint and a QoS constraint. We further derive a closed-form solution to optimal antenna placement with phase correction to enable coherent combining. Simulation results confirm data rate improvements while satisfying uRLLC requirements.


\section{System Model and Problem Formulation}
\subsection{System Model}

In this work, 
we consider a downlink communication system with a single waveguide equipped with $N$ pinching antennas, which serves a single-antenna mobile user. The user is assumed to be randomly located within a square region of side length $D$, with the position represented as $u=\left(x,y,0\right)$. The waveguide is aligned along the $x$-axis, with the antenna height set to $d$. It is assumed that $N$ pinching antennas are uniformly deployed along the waveguide at positions $\tilde{p}_{n}=\left(\tilde{x}_{n},0,d\right),\forall n\in\mathcal{N}\triangleq\left\{ 1,\cdots,N\right\}$, with adjacent antennas spaced by more than a minimum guard distance $\Delta$ to avoid electromagnetic coupling \cite{zhang2022beam}. Based on the spherical-wave propagation model, the channel vector between the set of antennas and the user is written as
\begin{equation}
    \begin{aligned}
\mathbf{h}_{\textrm{Pin}}=\left[a_{1},\cdots,a_{n},\cdots,a_{N}\right]^{T},
    \end{aligned}
\end{equation}
where $a_{n}=\frac{\alpha e^{-j\frac{2\pi}{\lambda}\left|u-\tilde{p}_{n}\right|-j\frac{2\pi}{\lambda_{g}}\left|\tilde{p}_{0}-\tilde{p}_{n}\right|}}{\left|u-\tilde{p}_{n}\right|},\forall n\in\left\{ 1,\cdots,N\right\}$.
 Here, $\alpha=\frac{c}{4\pi f_{c}}$ is a propagation constant, with $c$ being the speed of light, and $f_{c}$ being the carrier frequency. The waveguide wavelength is given by $\lambda_{g}=\frac{\lambda}{n_{\textrm{eff}}}$ with $\lambda $ being the free-space wavelength, and $n_{\textrm{eff}}$ being the effective refractive index of the waveguide. Each element in the channel vector incorporates two phase-shift terms, which are the free-space propagation delay $\varphi_{n}=\frac{2\pi}{\lambda}\left|u-\tilde{p}_{n}\right|$ and the waveguide-induced phase shift $\theta_{n}=\frac{2\pi}{\lambda_{g}}\left|\tilde{p}_{0}-\tilde{p}_{n}\right|$ with $\tilde{p}_{0}=\left(\tilde{x}_{0},0,d\right)$ denoting the location of the waveguide’s feed point.

Given the successive configuration of the antennas along the waveguide, we have the ordering $\tilde{x}_{n}>\tilde{x}_{n-1},\forall n\in\left\{ 2,\cdots,N\right\}$. Since the signal transmitted from any antenna is merely a phase-shifted version of that from another antenna, the system can only support a single data stream \cite{pozar2021microwave}. For simplicity, we assume that the total transmit power $P_{t}$ is evenly distributed across the $N$ pinching antennas \cite{ding2025flexible}. Hence, the received signal at the user is written as
\begin{align}
    y	 =\left(\sum_{n=1}^{N}\frac{\alpha e^{-j\varphi_{n}}e^{-j\theta_{n}}}{\left|u-\tilde{p}_{n}\right|}\right)\sqrt{\frac{P_{t}}{N}}s+\omega,
\end{align}
where $s$ denotes the transmitted signal along the waveguide, and $\omega\sim\mathcal{CN}\left(0,\sigma^{2}\right)$ represents additive complex Gaussian noise with variance $\sigma^{2}$. Then, the signal-to-noise ratio (SNR) is given by
\begin{align}
    \gamma	&=\frac{P_{t}}{N\sigma^{2}}\left|\sum_{n=1}^{N}\frac{\alpha e^{-j\varphi_{n}}e^{-j\theta_{n}}}{\left|u-\tilde{p_{n}}\right|}\right|^{2}=\frac{P_{t}\alpha^{2}}{N\sigma^{2}}\left|\sum_{n=1}^{N}\frac{e^{-j\left(\varphi_{n}+\theta_{n}\right)}}{\left|u-\tilde{p}_{n}\right|}\right|^{2}.
\end{align}

This work investigates finite blocklength communication.
Unlike Shannon's model, the strict latency constraints of uRLLC enforce short blocklength lengths $l$, rendering traditional capacity formulas inapplicable due to the fundamental trade-offs.
To address this limitation, a non-asymptotic framework for finite blocklength communication was proposed in \cite{polyanskiy2010channel}, which provides an accurate approximation of the achievable rate. Motivated by this, the maximum coding rate that can be achieved at a given $\gamma $, the decoding error probability $\varepsilon $, and the finite blocklength $ \mathit{l}$ is described by 
\begin{align}\label{rate}
    R\left(\gamma\right)=\ln\left(1+\gamma\right)-\tau\sqrt{1-\frac{1}{\left(1+\gamma\right)^{2}}},
\end{align}
where $\tau\triangleq\frac{Q^{-1}\left(\varepsilon\right)}{\sqrt{l}}$. $Q^{-1}\left(\varepsilon\right)$ denotes the inverse of the Gaussian $Q$-function, defined as $Q\left(x\right)=\frac{1}{\sqrt{2\pi}}\int_{x}^{\infty}e^{-\frac{t^{2}}{2}}dt$. This equation provides a more accurate estimate of the achievable rate in uRLLC with finite blocklength.

\subsection{Problem Formulation}
In this paper, we focus on a single-user uRLLC system with pinching antennas under the constraint of a minimum rate requirement, and the corresponding optimization problem is formulated as
\begin{align}
    \underset{\tilde{x}_{1},\cdots,\tilde{x}_{N}}{\textrm{max}}\quad	&R\left(\gamma\right),\label{P1}\\
\textrm{s.t.}	\quad&C_{1}:R\left(\gamma\right)\geq\frac{B}{l}\ln\left(2\right),\tag{5a}\\
	&C_{2}:\tilde{x}_{n}-\tilde{x}_{n-1}\geq\Delta,\forall n\in\left\{ 2,\cdots,N\right\}.\tag{5b}
\end{align}
Constraint $C_{1}$ guarantees that the user’s data rate satisfies the QoS requirement, with a target packet error probability of $10^{-5}$ for 32 bytes ($B=256$ bits), while $C_{2}$ enforces a minimum antenna spacing to prevent coupling between adjacent pinching antennas.
Meanwhile, problem \eqref{P1} is challenging due to two main factors. Firstly, the rate function under finite blocklength transmission is more complex than the Shannon capacity, complicating both the objective and QoS constraint. Secondly, the channel characteristics are affected by pinching antenna positions, further complicating the problem.

\section{The Position Design of Pinching Antennas}
In this section, we consider the position design of pinching antennas under finite blocklength transmission by solving problem \eqref{P1}. Specifically, we first equivalently transform the QoS constraint and the objective function by analyzing the data rate function under the finite blocklength transmission. Subsequently, based on the transformed optimization problem, we further optimize the positions of pinching antennas to derive a closed-form position solution by ignoring phase constraint, and then the phase alignment issue is further addressed by applying a fine-tuning algorithm to adjust antenna positions.

\subsection{Equivalent Transformation of Problem \eqref{P1}}
In this subsection, we first deal with the QoS constraint and the objective function in problem \eqref{P1}. In the following proposition, we present an equivalent transformation form of the QoS constraint.
\begin{prop}\label{prop1}
For problem \eqref{P1}, the constraint $R\left(\gamma\right)\geq\frac{B}{l}\ln\left(2\right)$ is equivalent to
\begin{align}\label{qos}
    \gamma\geq\nu_{2},
\end{align}
where $\nu_{2}=e^{\frac{\mathcal{W}\left(^{2\tau,-2\tau};-4\tau^{2}*2^{-\frac{2B}{l}}\right)}{2}+\frac{B}{l}\ln\left(2\right)}-1$. Here $\mathcal{W}$ is the generalized Lambert W function given by
\begin{align}
    \mathcal{W}\left(^{\iota_{1},\iota_{2}};\mu\right)	&=\iota_{1}-\sum_{m=1}^{\infty}\frac{1}{m\ast m!}\left(\frac{\mu me^{-\iota_{1}}}{\iota_{2}-\iota_{1}}\right)^{m}\\
	&\times\mathcal{B}_{m-1}\left(\frac{-2}{m\left(\iota_{2}-\iota_{1}\right)}\right)
\end{align}
 with $ \mathcal{B}_{m}\left(z\right)=\sum_{k=0}^{m}\frac{\left(m+k\right)!}{k!\left(m-k\right)!}\left(\frac{z}{2}\right)^{k}$.
\end{prop}
\begin{proof}
 Please refer to Appendix \ref{app1}.
\end{proof}
\begin{coro}\label{cor1}
    $R\left(\gamma\right)$ is monotonically increasing within $\left(\nu_{2},+\infty\right)$.
\end{coro}
As derived in Proposition \ref{prop1} and Corollary \ref{cor1}, the rate $R\left(\gamma\right)$ is monotonically increasing with respect to the SNR $\gamma$ under the QoS constraint. Leveraging this property, the objective function can be equivalently transformed into the form stated in Proposition \ref{prop2}.
\begin{prop}\label{prop2}
For problem \eqref{P1}, the objective function can be transformed into
\begin{align}
    \underset{\tilde{x}_{1},\cdots,\tilde{x}_{N}}{\textrm{max}}\quad
    \left|\sum_{n=1}^{N}\frac{e^{-j\psi_{n}}}{\sqrt{\left(\tilde{x}_{n}-x\right)^{2}+y^{2}+d^{2}}}\right|,
\end{align}
where $\psi_{n}\triangleq\varphi_{n}+\theta_{n}=2\pi\left(\frac{\left|u-\tilde{p}_{n}\right|}{\lambda}+\frac{\left|\tilde{p}_{0}-\tilde{p}_{n}\right|}{\lambda_{g}}\right)$.
\end{prop}
\begin{proof}
 Please refer to Appendix \ref{app2}.
\end{proof}
In order to maximize the objective function, we assume that the phase term $e^{-j\psi_{n}}$ is aligned modulo $2\pi$. Based on Propositions \ref{prop1} and \ref{prop2}, the optimization problem can be transformed into
\begin{align}\label{P2}
    \underset{\tilde{x}_{1},\cdots,\tilde{x}_{N}}{\textrm{max}}	\quad&\sum_{n=1}^{N}f\left(\tilde{x}_{n}\right),\\
\textrm{s.t.}\quad	&C_{1}':\frac{P_{t}\alpha^{2}}{N\sigma^{2}}\left|\sum_{n=1}^{N}f\left(\tilde{x}_{n}\right)\right|^{2}\geq\nu_{2},	C_{2}, C_{3},\tag{10a} 
\end{align}
where 
$f\left(\tilde{x}_{n}\right)\triangleq\frac{1}{\left|u-\tilde{p}_{n}\right|}=\frac{1}{\sqrt{\left(\tilde{x}_{n}-x\right)^{2}+C}}$ 
with $C=y^{2}+d^{2}>0$ and 
   $ \psi_{n}=2\pi\left(\frac{\sqrt{\left(\tilde{x}_{n}-x\right)^{2}+C}}{\lambda}+\frac{{\left|\tilde{x}_{n}-\tilde{x}_{0}\right|}}{\lambda_{g}}\right)$.

Despite simplifying the complex rate function through this transformation, the position optimization problem remains challenging due to the non-convexity of both the objective function and the QoS constraint.
In the next subsection, we proceed to solve this problem in detail.

\subsection{Positioning Optimization for Pinching Antennas}
The main challenge of this problem lies in the joint optimization of phase shifts and antenna positions. To address this, we adopt a step-by-step optimization strategy. Firstly, we ignore phase constraint to optimize positions only, then adjust phase shifts based on the optimized position results.

After ignoring the phase constraint, problem \eqref{P2} is transformed into
\begin{align}\label{P3}
    \underset{\tilde{x}_{1},\cdots,\tilde{x}_{N}}{\textrm{max}}	\quad&\sum_{n=1}^{N}f\left(\tilde{x}_{n}\right),\\
\textrm{s.t.}\quad &C_{1}', C_{2}.\tag{11a}
\end{align}
Based on the symmetry and monotonicity of $f\left(\tilde{x}_{n}\right)$, we derive the following closed-form expression for the optimal positions of pinching antennas.
\begin{prop}\label{prop3}
    For problem \eqref{P3}, the optimal solution is
    \begin{align}
        \tilde{x}_{n}^{*}=x+\left(n-\frac{N+1}{2}\right)\Delta.
    \end{align}    
\end{prop}
\begin{proof}
     Please refer to Appendix \ref{app3}.
\end{proof}

This proposition derives the optimal closed-form solution for antenna positions when ignoring phase constraint. To ensure the problem is solvable, we propose the following lemma.
\begin{lem}\label{lem1}
    The feasibility of problem \eqref{P3} is guaranteed by 
    \begin{align}\label{feasibility}
        C<\frac{NP_{t}\alpha^{2}}{\nu_{2}\sigma^{2}}.
    \end{align}
\end{lem}
\begin{proof}
    Please refer to Appendix \ref{app4}.
\end{proof}

Subsequently, by incorporating the phase constraint $C_{3}$, only minor adjustments to the positions of pinching antennas are required. The handling phase constraint
\begin{align}
    \psi_{n}-\psi_{n-1}=2k\pi,\forall n\in\left\{ 2,\cdots,N\right\} 
\end{align}
can be converted into
\begin{align}
    \textrm{mod}\left\{  \Psi\left(\tilde{x}_{n}\right)-\Psi\left(\tilde{x}_{n-1}\right),2\pi\right\}=0,
\end{align}
where $\textrm{mod}\left\{ a,b\right\} $ denotes the modulo operation of $a$ with respect to $b$.
To address this constraint, a per-antenna fine-tuning algorithm is proposed to determine the antenna positions under the phase constraint. This procedure is summarized in Algorithm \ref{alg1}, which performs a forward sweep to adjust antennas to the right and a backward sweep to adjust those to the left. The algorithm ensures that the phase difference between adjacent antennas is an integer multiple of $2\pi$. At each step, a local modular equation is solved to iteratively refine the antenna positions while simultaneously satisfying both the QoS and the minimum spacing requirement.

\begin{algorithm}
\caption{Per-antenna Fine-tuning Algorithm}\label{alg1}
\begin{algorithmic}[1]
\REQUIRE Number of antenna $N$, minimum antenna spacing $\Delta$, initial location of center antenna $\tilde{x}_{\frac{N+1}{2}}=x$, phase function $\Psi(\tilde{x}_{n})$, $\Psi_{\text{prev}} = \Psi(\tilde{x}_{n})$, $\Psi_{\text{next}} = \Psi(\tilde{x}_{n+1})$
\ENSURE Optimal position of pinching antennas $ \tilde{x}^{*}_{1}, \cdots, \tilde{x}^{*}_{N}$

\STATE \textbf{Forward pass:} refine antennas $\frac{N+1}{2}+1,\cdots,N$
\FOR{$n = \frac{N+1}{2}+1$ \TO $N$}
    \STATE $a \gets \tilde{x}_{n-1} + \Delta$, $b \gets \tilde{x}_{n-1} + 3 \times \Delta$
    \STATE Solve $\textrm{mod}\left\{  \Psi\left(\tilde{x}_{n}\right)-\Psi_{\text{prev}},2\pi\right\}=0$  for $\tilde{x}_{n} \in [a, b]$ to get $\tilde{x}^{*}_{n}$
    \STATE $\Psi_{\text{prev}} \gets \Psi(\tilde{x}^{*}_{n})$
\ENDFOR

\STATE \textbf{Backward pass:} refine antennas $\frac{N+1}{2}-1,\cdots,1$
\FOR{$n = \frac{N+1}{2}-1$ to $1$}
\STATE $c \gets \tilde{x}_{n+1} - \Delta$, $d \gets \tilde{x}_{n+1} - 3 \times \Delta$
\STATE Solve $\textrm{mod}\left\{  \Psi_{\text{next}}-\Psi\left(\tilde{x}_n\right),2\pi\right\}=0$  for $\tilde{x}_{n} \in [d, c]$ to get $\tilde{x}^{*}_{n}$ 
\STATE$\Psi_{\text{next}} \gets \Psi(\tilde{x}^{*}_{n})$
\ENDFOR
\RETURN $ \tilde{x}^{*}_{1}, \cdots, \tilde{x}^{*}_{N}$
\end{algorithmic}
\end{algorithm}
\begin{rem}
    When $N$ is an even number, Algorithm \ref{alg1} remains valid simply by replacing $\frac{N+1}{2}$ in Algorithm \ref{alg1} with $\frac{N}{2}$.
\end{rem}

\section{Simulation Results}

In this section, we verify the effectiveness of the proposed scheme through numerical simulations.
All simulations are carried out using a free-space wavelength of $\lambda = \frac{c}{f_c}$, corresponding to a carrier frequency of $f_c = 28~\textrm{GHz}$. The antenna height is set to $d = 3~\textrm{m}$, the noise power is assumed to be $\sigma^2 = -90~\textrm{dBm}$, and the antenna spacing is chosen as $\Delta = \frac{\lambda}{2}$. The waveguide has an effective refractive index of $n_{\textrm{eff}} = 1.4$, and the feed point is located at the coordinate $(0, 0, d)$.

\begin{figure}[htbp]
\centering
\includegraphics[scale=0.45]{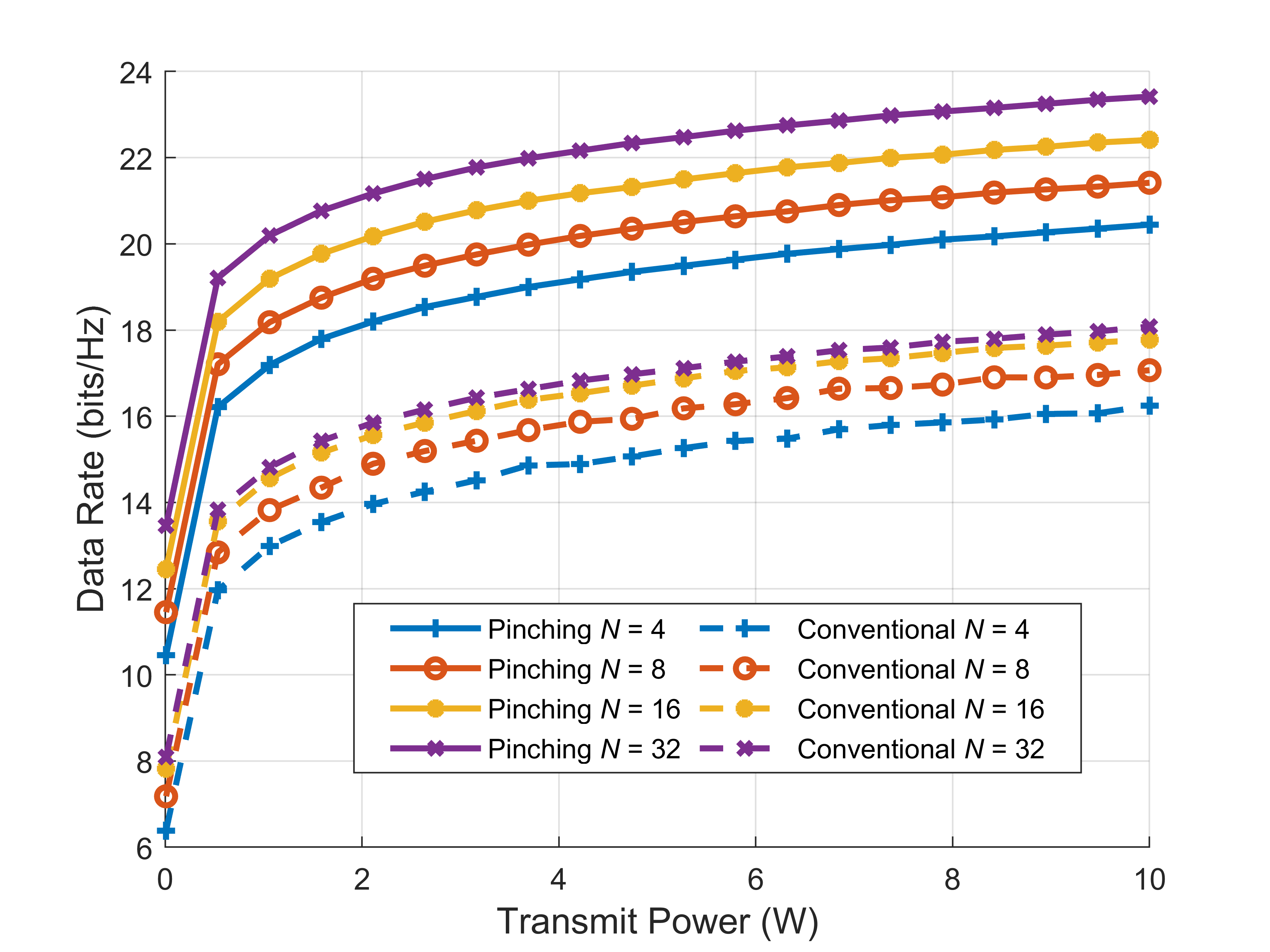}
\caption{Data rate versus transmit power for cases of system with pinching antennas and conventional antennas.}
\label{fig2}
\end{figure}

Fig. \ref{fig2} studies the achievable data rate versus the transmit power under different number of antennas $N$. Both the pinching antenna system and the conventional antenna system exhibit data rate saturation at high transmit power. 
This saturation occurs because, under finite blocklength constraint, increasing transmit power leads to a diminishing trend in rate improvement due to the non-negligible dispersion term in the rate expression. 
Compared with the system using conventional antennas, the system using pinching antennas achieves higher data rates. As expected, increasing the number of antennas improves the data rate of the system. Moreover, with a larger number of antennas, the advantage of pinching antenna becomes more pronounced over the conventional one.



\begin{figure}[htbp]
\centering
\includegraphics[scale=0.45]{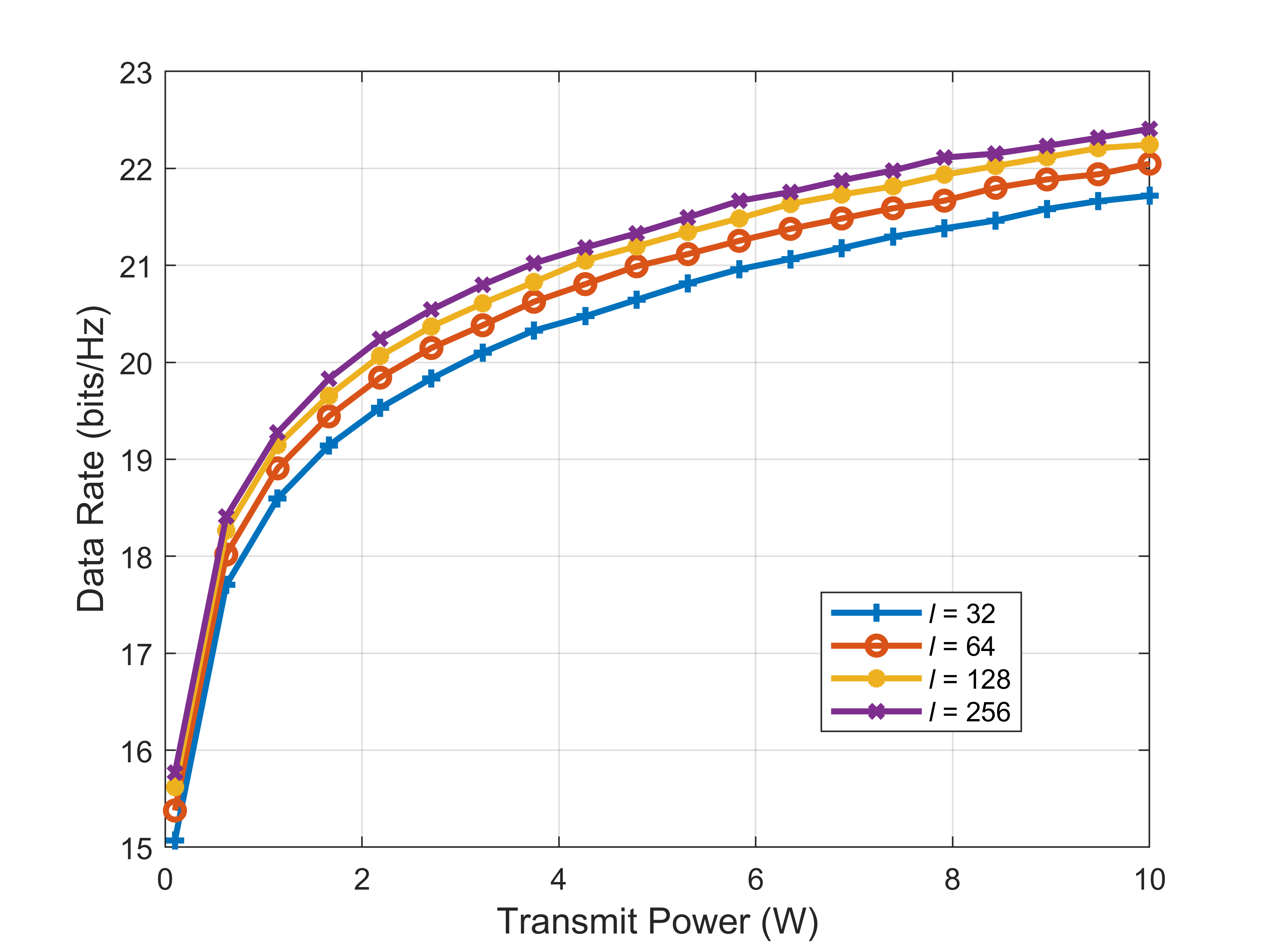}
\caption{Data rate versus transmit power for different QoS constraints.}
\label{fig3}
\end{figure}

Fig. \ref{fig3} shows the relationship between transmit power and the achievable data rate under various blocklengths $l$. As transmit power increases, the data rate improves across all blocklength settings, with longer blocklength yielding higher data rates. These results demonstrate that extending blocklength can effectively mitigate the rate loss caused by the strict reliability and latency constraints of uRLLC systems. 

\begin{figure}[htbp]
\centering
\includegraphics[scale=0.45]{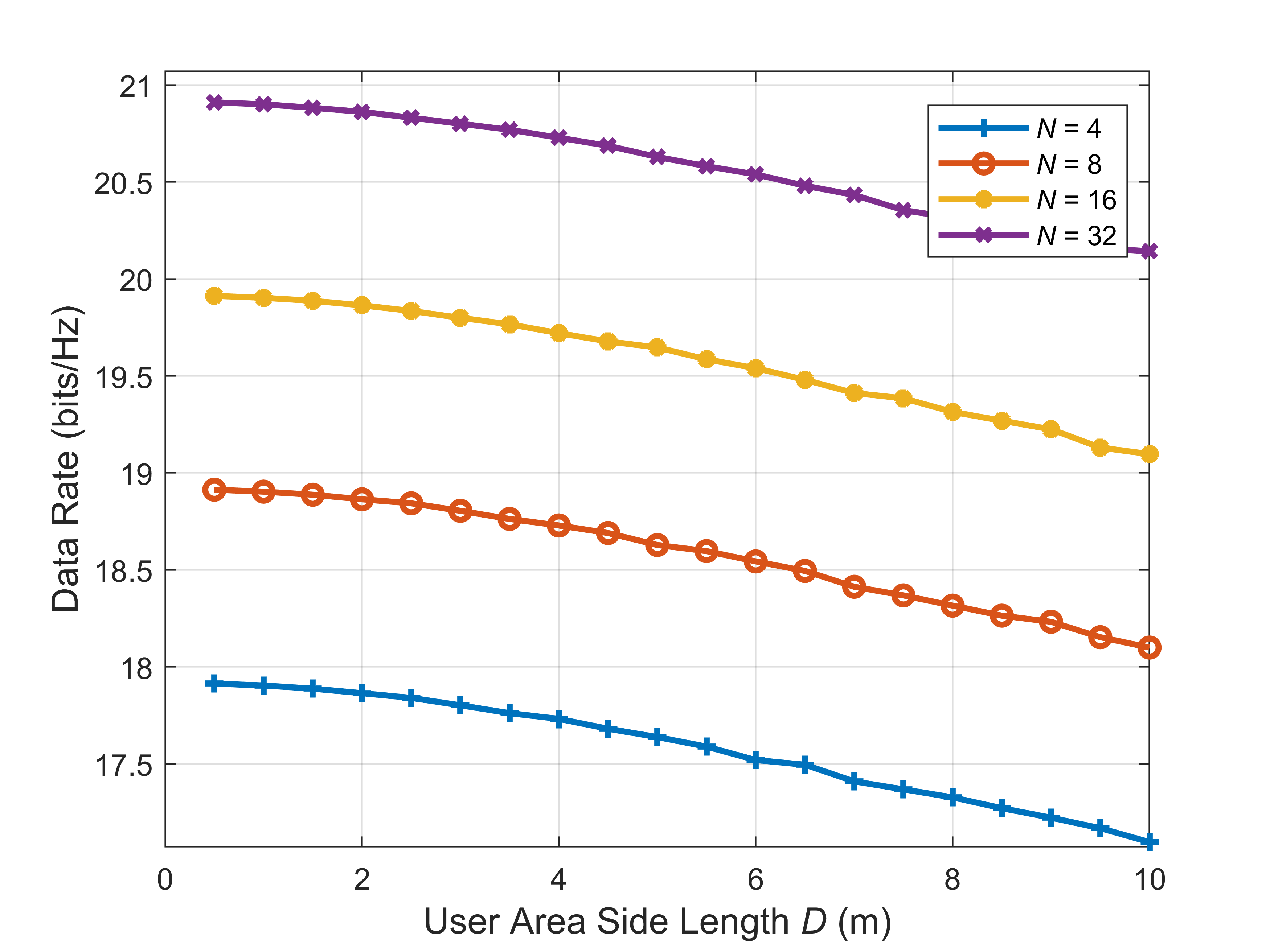}
\caption{Data rate versus user area side length.}
\label{fig4}
\end{figure}

Fig. \ref{fig4} illustrates how user area side length $D$ affects data rate for various antenna numbers $N$. As $D$ increases, one can observe that the system data rate decreases, with the decline becoming more pronounced for larger $N$. This trend matches theoretical expectations, i.e., a larger user area will cause the user to be farther from the feed point, leading to higher path loss, highlighting that system performance is sensitive to user distribution, particularly with fewer pinching antennas.

\section{Conclusion}
In this work, we proposed an antenna position design for pinching antennas in a single-user uRLLC system. By optimizing antenna positions and applying phase correction, we derived a closed-form algorithmic solution for the optimal placement of pinching antennas. Simulation results demonstrated that increasing the number of pinching antennas significantly improves the achievable data rate. Future work will focus on extending the proposed model to multi-user scenarios.

{\appendices
\section{Proof of the Proposition \ref{prop1}}\label{app1}
In order to handle the QoS constraint, we first analyze the function \eqref{rate}. Taking the first derivative of the function $R\left(\gamma\right)$ and setting it to zero, we have 
\begin{align}\label{rate'}
    R'\left(\gamma\right)=\frac{1}{1+\gamma}\left(1-\tau\frac{1}{\left(1+\gamma\right)\sqrt{\left(1+\gamma\right)^{2}-1}}\right)=0,
\end{align}
yielding the solution $\nu_{0}=\sqrt{\frac{1+\sqrt{1+4\tau^{2}}}{2}}-1$. 
Obviously, $R'\left(\gamma\right)$ is monotonically increasing with respect to $\gamma$ with $\tau=\frac{Q^{-1}\left(\varepsilon\right)}{\sqrt{l}}>0$.
As a result, $R'\left(\gamma\right)<R'\left(\nu_0\right)=0$ when $0<\gamma<\nu_{0}$ and $R'\left(\gamma\right)>R'\left(\nu_0\right)=0$ when $\gamma>\nu_{0}$, i.e., $R\left(\gamma\right)$ is monotonically decreasing in $0<\gamma<\nu_{0}$ and monotonically increasing in $\gamma>\nu_0$. Correspondingly, $R\left(\nu_{0}\right)<R\left(0\right)=0$.

Then, we further derive the solution to $R\left(\gamma\right)=\beta$ as
\begin{align}
    \gamma=e^{\frac{K}{2}+\beta}-1=e^{\frac{\mathcal{W}\left(^{2\tau,-2\tau};-4e^{-2\beta}\tau^{2}\right)}{2}+\beta}-1,
\end{align}
where $\mathcal{W}$ is given by
\begin{align}
    \mathcal{W}\left(^{\iota_{1},\iota_{2}};\mu\right)	&=\iota_{1}-\sum_{m=1}^{\infty}\frac{1}{m\ast m!}\left(\frac{\mu me^{-\iota_{1}}}{\iota_{2}-\iota_{1}}\right)^{m}\\
	&\times\mathcal{B}_{m-1}\left(\frac{-2}{m\left(\iota_{2}-\iota_{1}\right)}\right)
\end{align}
 with $\mathcal{B}_{m}\left(z\right)=\sum_{k=0}^{m}\frac{\left(m+k\right)!}{k!\left(m-k\right)!}\left(\frac{z}{2}\right)^{k}$.
For $\beta=0$ and $R\left(\gamma\right)=\frac{B}{l}\ln\left(2\right)$, the corresponding solutions are
\begin{align}
\nu_{1}=e^{\frac{\mathcal{W}\left(^{2\tau,-2\tau};-4\tau^{2}\right)}{2}}-1,
\end{align}
and
\begin{align}
\nu_{2}=e^{\frac{\mathcal{W}\left(^{2\tau,-2\tau};-4\tau^{2}*2^{-\frac{2B}{l}}\right)}{2}+\frac{B}{l}\ln\left(2\right)}-1,
\end{align}
respectively.

Since $R\left(\nu_{0}\right)<0$ and $R\left(\nu_{1}\right)=0$, we have $\nu_1>\nu_{0}$. Since the Lambert W function is monotonically increasing, and $2^{-\frac{2B}{l}}<1$ holds for all $B>0$ and $l>0$, it follows that $\nu_{2}>\nu_{1}>\nu_{0}$, ensuring $R\left(\gamma\right) $ is monotonically increasing in $\left(\nu_{2},+\infty\right)$. Hence, the QoS constraint can be simplified to \eqref{qos}.

\section{Proof of the Proposition \ref{prop2}}\label{app2}
The expression of SNR is $\gamma=\frac{P_{t}\alpha^{2}}{N\sigma^{2}}\left|\sum_{n=1}^{N}\frac{e^{-j\psi_{n}}}{\left|u-\tilde{p}_{n}\right|}\right|^{2}$, where $\psi_{n}\triangleq\varphi_{n}+\theta_{n}=2\pi\left(\frac{\left|u-\tilde{p}_{n}\right|}{\lambda}+\frac{\left|\tilde{p}_{0}-\tilde{p}_{n}\right|}{\lambda_{g}}\right)$. Since $R\left(\gamma\right)$ is monotonically increasing for $\gamma\geq\nu_{2}$,
\begin{align}
    \textrm{max}\quad R\left(\gamma\right)\iff\textrm{max}\quad\left|\sum_{n=1}^{N}\frac{e^{-j\psi_{n}}}{\left|u-\tilde{p}_{n}\right|}\right|.
\end{align}
As a result, we have Proposition \ref{prop2}.
\section{Proof of the Proposition \ref{prop3}}\label{app3}
Here, we use proof by contradiction to demonstrate Proposition \ref{app3}. We assume that \( \tilde{x}_n^* - \tilde{x}_{n-1}^* > \Delta \) for some \( n \). If \( \tilde{x}_n^* + \tilde{x}_{n-1}^* < 2x \), move the \((n{-}1)\)-th antenna toward the \(n\)-th by \( \delta_1 > 0 \) to satisfy \( \tilde{x}_n^* - (\tilde{x}_{n-1}^* + \delta_1) = \Delta \). In this case, $\sqrt{(\tilde{x}_n - x)^2 + C}$ is reduced, leading to an increase in the objective function. Similarly, if \( \tilde{x}_n^* + \tilde{x}_{n-1}^* \geq 2x \), shifting the \(n\)-th antenna left by \( \delta_2 \) achieves the same. In both cases, the modified spacing yields a better solution. Thus, \( \tilde{x}_n^* - \tilde{x}_{n-1}^* = \Delta \) must hold for all \( n \).

With uniform spacing, let \( \tilde{x}_n = \tilde{x}_1 + (n - 1)\Delta \). The objective becomes
\begin{align}
    \max_{\tilde{x}_1}\quad &\ g(\tilde{x}_1) = \sum_{n=1}^{N} \frac{1}{\sqrt{[\tilde{x}_1 + (n - 1)\Delta - x]^2 + C}},
\\
\text{s.t.}\quad &\gamma=\frac{P_{t}\alpha^{2}}{N\sigma^{2}}\left[g\left(\tilde{x}_{1}\right)\right]^{2}\geq\nu_{2}\iff g(\tilde{x}_1) \geq K,
\end{align}
where $K=\sqrt{\frac{\nu_{2}N\sigma^{2}}{P_{t}\alpha^{2}}}$ since $g\left(\tilde{x}_1\right)>0$. It is known that $g\left(\tilde{x}_1\right)$ is strictly symmetric about $\tilde{x}_1^* = x - \frac{N - 1}{2} \Delta$, with $g'\left(\tilde{x}_1^*\right)=0$ and $g''\left(\tilde{x}_1^*\right)<0$. Combining symmetry and properties of the first and second derivatives, $\tilde{x}_1^*$ is the unique global maximum point of $g\left(\tilde{x}_1\right)$. Then the optimal solution to problem \eqref{P3} is
\begin{align}
    \tilde{x}_{n}^{*}=\tilde{x}_{1}^*+\left(n-1\right)\Delta=x+\left(n-\frac{N+1}{2}\right)\Delta.
\end{align}

\section{Proof of the Lemma \ref{lem1}}\label{app4}
To ensure $g\left(\tilde{x}_{1}^*\right)\geq K$, $\tilde{x}_{n}^{*}$ still needs to satisfy
\begin{align}\label{ineq}
    \sum_{n=1}^{N}\frac{1}{\sqrt{\left(\tilde{x}_{n}^{*}-x\right)^{2}+C}}\geq\sqrt{\frac{\nu_{2}N\sigma^{2}}{P_{t}\alpha^{2}}}.
\end{align}
The summation form poses computational challenges. To simplify the analysis, we ensure that the maximum term in the summation satisfies the condition, which is the necessary condition for the above inequality. By observation, we note that the farther $\tilde{x}_{n}^{*}$ is from $x$, the smaller $\frac{1}{\left(\tilde{x}_{n}^{*}-x\right)+C}$ becomes. Since the $\tilde{x}_{n}^{*}$ values are symmetrically distributed around $x$, $\frac{1}{\left(\tilde{x}_{n}^{*}-x\right)+C}$ attains its maximum when $\tilde{x}_{n}^{*}=x$, at which point
\begin{align}
    \sum_{n=1}^{N}\frac{1}{\sqrt{\left(\tilde{x}_{n}^{*}-x\right)^{2}+C}}<\frac{N}{\sqrt{C}}.
\end{align}
Thus, the feasibility of problem \eqref{P3} is guaranteed by \eqref{feasibility}.
}

 \bibliographystyle{IEEEtran}
\bibliography{reference.bib}
%


 





\end{document}